\documentclass[aps,pra,onecolumn,superscriptaddress,groupedaddress]{revtex4}  
\usepackage{graphicx}  
\usepackage{dcolumn}   
\usepackage{bm}        
\usepackage{amsthm}
\usepackage{amsmath}
\usepackage{amssymb}
\usepackage{amsfonts}
\usepackage{amscd}
\usepackage{enumerate}
\usepackage{bm}
\usepackage{latexsym}
\usepackage{color}
\usepackage{hyperref}


\def\;{{\hspace{0.3ex};\hspace{0.5ex}}}
\def\,{{\hspace{0,3ex},\hspace{0.5ex}}}
\def\({{\hspace{1.2ex}(}}

\def\rank{{\rm R}}

\def\C{{\mathbb{C}}}

\def\dim{{{\rm dim}}}
\def\rank{{{\rm rank}}}

\def\QED{\mbox{\rule[0pt]{1.5ex}{1.5ex}}}

\def\endproof{\hspace*{\fill}~\QED\par\endtrivlist\unskip}

\newtheorem{definition}{Definition}
\newtheorem{proposition}[definition]{Proposition}
\newtheorem{lemma}[definition]{Lemma}

\newtheorem{theorem}[definition]{Theorem}
\newtheorem{corollary}[definition]{Corollary}
\newtheorem{conjecture}[definition]{Conjecture}

\newtheorem{remark}[definition]{Remark}
\newtheorem{example}[definition]{Example}
\newtheorem{question}[definition]{Question}
\newtheorem{opprob}[definition]{Open Problems}

\newcommand{\nc}{\newcommand}

 \nc{\cBA}{\mathbb{A}} \nc{\cBB}{\mathbb{B}} \nc{\cBC}{\mathbb{C}}
 \nc{\cBD}{\mathbb{D}} \nc{\cBE}{\mathbb{E}} \nc{\cBF}{\mathbb{F}}
 \nc{\cBG}{\mathbb{G}} \nc{\cBH}{\mathbb{H}} \nc{\cBI}{\mathbb{I}}
 \nc{\cBJ}{\mathbb{J}} \nc{\cBK}{\mathbb{K}} \nc{\cBL}{\mathbb{L}}
 \nc{\cBM}{\mathbb{M}} \nc{\cBN}{\mathbb{N}} \nc{\cBO}{\mathbb{O}}
 \nc{\cBP}{\mathbb{P}} \nc{\cBQ}{\mathbb{Q}} \nc{\cBR}{\mathbb{R}}
 \nc{\cBS}{\mathbb{S}} \nc{\cBT}{\mathbb{T}} \nc{\cBU}{\mathbb{U}}
 \nc{\cBV}{\mathbb{V}} \nc{\cBW}{\mathbb{W}} \nc{\cBX}{\mathbb{X}}
 \nc{\cBZ}{\mathbb{Z}}
 \nc{\trans}{^\top}

\def\min{\mathop{\rm min}}

\nc{\cA}{{\cal A}} \nc{\cB}{{\cal B}} \nc{\cC}{{\cal C}}
\nc{\cD}{{\cal D}} \nc{\cE}{{\cal E}} \nc{\cF}{{\cal F}}
\nc{\cG}{{\cal G}} \nc{\bH}{{\cal H}} \nc{\cI}{{\cal I}}
\nc{\cJ}{{\cal J}} \nc{\cK}{{\cal K}} \nc{\cL}{{\cal L}}
\nc{\cM}{{\cal M}} \nc{\cN}{{\cal N}} \nc{\cO}{{\cal O}}
\nc{\cP}{{\cal P}} \nc{\cQ}{{\cal Q}} \nc{\cR}{{\cal R}}
\nc{\cS}{{\cal S}} \nc{\cT}{{\cal T}} \nc{\cU}{{\cal U}}
\nc{\cV}{{\cal V}} \nc{\cW}{{\cal W}} \nc{\cX}{{\cal X}}
\nc{\cZ}{{\cal Z}} \nc{\cY}{{\cal Y}}

\def\a{\alpha}
\def\b{\beta}

\def\ba{\mathbf{a}}
\def\cBb{\mathbf{b}}
\def\cBc{\mathbf{c}}
\def\cBd{\mathbf{d}}

\def\be{\mathbf{e}}

\def\bm{\mathbf{m}}
\def\bn{\mathbf{n}}
\def\bu{\mathbf{u}}
\def\bv{\mathbf{v}}
\def\bx{\mathbf{x}}
\def\by{\mathbf{y}}
\def\bz{\mathbf{z}}
\def\bH{\mathbf{H}}
\def\bU{\mathbf{U}}
\def\bW{\mathbf{W}}

\def\u{\upsilon}

\def\ox{\otimes}

\newcommand{\red}{\textcolor{red}}

\newcommand{\ket}[1]{|#1\rangle}

\newcommand{\abs}[1]{|#1|}

\nc{\U}{\mathrm{U}}

\def\bcj{\begin{conjecture}}
\def\ecj{\end{conjecture}}
\def\bcr{\begin{corollary}}
\def\ecr{\end{corollary}}
\def\bd{\begin{definition}}
\def\ed{\end{definition}}
\def\bea{\begin{eqnarray}}
\def\eea{\end{eqnarray}}
\def\bem{\begin{enumerate}}
\def\eem{\end{enumerate}}
\def\bex{\begin{example}}
\def\eex{\end{example}}
\def\bim{\begin{itemize}}
\def\eim{\end{itemize}}
\def\bl{\begin{lemma}}
\def\el{\end{lemma}}
\def\bma{\begin{bmatrix}}
\def\ema{\end{bmatrix}}
\def\bpf{\begin{proof}}
\def\epf{\end{proof}}
\def\bpp{\begin{proposition}}
\def\epp{\end{proposition}}
\def\bqu{\begin{question}}
\def\equ{\end{question}}
\def\br{\begin{remark}}
\def\er{\end{remark}}
\def\bt{\begin{theorem}}
\def\et{\end{theorem}}

\begin{document}


\title{The tensor rank of tensor product of two three-qubit W states is eight}

\author{Lin Chen}\email{linchen@buaa.edu.cn}
\affiliation{School of Mathematics and Systems Science, Beihang University, Beijing 100191, China}
\affiliation{International Research Institute for Multidisciplinary Science, Beihang University, Beijing 100191, China}
 \author{Shmuel Friedland}\email{friedlan@uic.edu}
 \affiliation{Department of Mathematics, Statistics and Computer Science,
 University of Illinois at Chicago, Chicago, Illinois 60607-7045, USA}

\date{\today}

\begin{abstract}
\red{We show that the tensor rank of tensor product of two three-qubit W states is not less than eight.  Combining this result with the recent result of M. Christandl, A. K. Jensen, and J. Zuiddam that the tensor rank of tensor product of two three-qubit W states is at most eight, we deduce that the tensor rank of tensor product of two three-qubit W states is eight.} We also construct the upper bound of the tensor rank of tensor product of many three-qubit W states.
\end{abstract}




\maketitle




\section{Introduction}
\label{sec:int}
Let $\bH$ be an $n$-dimensional Hilbert space. 
We denote by a bold letter $\bx$ an element in $\bH$. 
For compactness of the exposition we adopt the following terminology.
A nonzero vector $\bx$ is called a state, while a normalized state is  a vector $\bx$ of norm one. 
For a positive integer $d>1$ a $d$-partite state is the Hilbert space $\bH=\bH_1\otimes\cdots\otimes\bH_d$, where $\dim\ \bH_i=n_i$ for $i\in[d]=\{1,\ldots,d\}$.  We denote $\otimes_{i=1}^d \bH_i=\bH$.
 In the case $\bH_1=\ldots=\bH_d$ we denote $\bH$ by $\otimes^d\bH_1$.
 An unentangled state is a rank one tensor  $\bx_1\otimes\cdots\otimes\bx_d$, where $\bx_i\ne \mathbf{0}, i\in[d]$.  We denote by a calligraphic  letter $\cX$ an element of $\otimes_{i=1}^d \bH_i$.  The rank of a state $\cX$, denoted by $\rank\ \cX$, is the minimal number $r$ in the decomposition of $\cX$ as a sum of  unentangled states
$\cX=\sum_{j=1}^r \otimes_{i=1}^d \bx_{i,j}$.  Thus $\rank\ \cX$ is a measurement of entanglement of a state.  There are other measure of entanglement of normalized states, as geometrical measure of entanglement \cite{wg2003,cxz2010} or the nuclear norm of $\cX$ \cite{dfl17}.

The entanglement of bipartite states, i.e. $d=2$, is well understood, since $\bH_1\otimes \bH_2$ can be identified with the space of $\dim\ \bH_1 \times \dim\ \bH_2$ matrices.  In this case $\rank\ \cX$ is the rank of the corresponding matrix, and the maximal value of this rank is $\min(\dim\ \bH_1, \dim\ \bH_2)$.  To emphasize that we are dealing with bipartite states, i.e. matrices, we will usualy denote by $X$ the matrix representing the bipartite state.
The first interesting case is the $3$-qubit states: $d=3, \dim\ \bH_1=\dim\ \bH_2=\dim\ \bH_3=2$.  There are two kinds of entangled states which can not be decomposed as a product of an unentangled state with a two qubit entangled state:  the $GHZ$ and $W$ states whose ranks are $2$ and $3$ respectively.
The closure of the orbit of $GHZ$ under the action of $GL(\C^2)\times GL(\C^2)\times GL(\C^2)$ is $\otimes^3\bH_1$, and its rank is two. The $W$ state has the maximum rank three.  We will usually denote the $W$ state by the tensor $\cW$.

We now consider another $d'$ partite state Hilbert space $\bH'=\otimes_{i'=1}^{d'} \bH_{i'}'$, where $\dim\ \bH_{i'}'=n'_{i'}, i'\in[d']$. We define two different tensor products of $\bH$ and $\bH'$.  The first product is the tensor product $\bH\otimes\bH'$.
It has the following physical interpretation.  The $d$ and $d'$ partite tensor products $\bH$ and $\bH'$ correspond to two sets of parties $\{P_1,\ldots,P_d\}$ and $\{Q_1,\ldots,Q_{d'}\}$.  Then $\bH\otimes\bH'$ corresponds to $d+d'$ party $\{P_1,\ldots,P_d,Q_1,\ldots,Q_{d'}\}$.  The second tensor product, which we call the Kronecker product, is defined as follows.  Assume that $d\le d'$.  (We can always achieve \red{this} by permuting the factors $\bH$ and $\bH'$.)  Then 
\[\bH\otimes_{K}\bH'=(\otimes_{i=1}^d (\bH_i\otimes \bH_i'))\otimes (\otimes_{i'=d+1}^{d'} \bH_{i'}').\] 
(If $d'=d$ the second tensor product is omitted.)  
The physical interpretation of the Kronecker product is as follows. The $d$ and $d'$ partite tensor products $\bH$ and $\bH'$ correspond to two sets of parties $\{P_1,\ldots,P_d\}$ and $\{P_1,\ldots,P_{d'}\}$ respectively.  Then $\bH\otimes_{K}\bH'$ corresponds to the party$\{P_1,\ldots,P_{d'}\}$ where each person $P_i$ has the space $\bH_i\otimes\bH_i'$ for $i\in[d]$.  For $i'>d$ the person $P_{i'}$ has the space $\bH_{i'}'$. Note that for $d=d'=2$ $\bH\otimes_K\bH'$ corresponds to the Kronecker product two matrix spaces.  Suppose that $\bH'=\bH$.  Then $\otimes^p \bH=\bH^{\otimes p}$ is $pd$ partite system corresponding to $p$ tensor products of $\bH$.  Furthermore, $\otimes^p_K \bH=\otimes_{i=1}^d(\otimes^p \bH_i)$.

Assume that $\cX\in\bH,\cY\in\bH'$ are two states.   
That is, the parties $\{P_1,\ldots,P_d\}$ and $\{Q_1,\ldots,Q_{d'}\}$ each share the state $\cX$ and $\cY$ respectively.  Then the two parties together share the state $\cX\otimes \cY$.
The rank of $\cX\otimes \cY$ is $\rank\ \cX\otimes \cY$.  Clearly,
$\rank\ \cX\otimes \cY\le (\rank\ \cX)(\rank\ \cY)$.  The tensor $\cX\otimes_K\cY\in \bH\otimes_K\bH'$ has the following physical interpretation.  In the party $\{P_1,\ldots,P_{d'}\}$
the person $P_i$ has part $i$ of $\cX$ and $\cY$ for $i\in[d]$, while the person $P_{i'}$
has only part $i'$ of $\cY$ for $i'>d$.  It is straightforward to show that $\rank\ \cX\otimes \cY\ge \rank\ \cX\otimes _K\cY$.  In particular, 
$\rank\ \cX\otimes_K\cY\le (\rank\ \cX)(\rank \cY)$.

Assume that $\bH=\bH'$.  Then $\otimes^p\cX=\cX^{\otimes p}\in \otimes^p\bH$ and $\otimes^p_K\cX\in\otimes_K^p\bH$.   Thus we have the inequalities
\begin{equation}\label{powrankin}
\rank\ \cX\otimes_K\cY\le\rank\ \cX\otimes \cY\le (\rank\ \cX)(\rank\ \cY), \quad\rank\ \otimes^p_K\cX\le \rank\ \otimes^p\cX\le (\rank\ \cX)^p.
\end{equation} 

These notions and operations have been applied to various problems in quantum information theory such as the conversion of multipartite state \red{\cite{cds08},\cite{ycg10},\cite{cds10},\cite{ccd2010}.  In these papers the authors consider the rank of tensors under the Kronecker product, (which they call the rank of the tensor product). It is shown in \cite{cds08} that 
\[\rank\ \cW\otimes_K\cW\le 8 < (\rank\ \cW)^2=9.\] 
That is, unlike for tensor product of matrices, the tensor rank is not multiplicative under the tensor Kronecker product. In \cite{ycg10},\cite{cds10} it is shown that $\rank\ \cW\otimes_K\cW= 7$.} Very recently it has been proved that the tensor rank is also not multiplicative under the tensor product \cite{cjz17}. In particular, authors in \cite{cjz17} have shown that the tensor product of two three-qubit W states has tensor rank at most eight. In this note we show that it is exactly eight.



The rest of this paper is organized as follows.  In Section \ref{sec:prelim} we explain our notations and recall some well known results for the rank of $3$-tensors.  First we recall Kruskal's theorem which gives a sufficient condition for uniqueness of rank decomposition of $3$-tensor \cite{KRUSKAL197795}.  Second we recall   \emph{Strassen's direct sum conjecture} \cite{Strassen1973Vermeidung}.  A special case of this conjecture was proven by Ja'Ja'-Takche \cite{jaja86}.  We state a restricted version of Strassen's conjecture and prove it in special cases using the results of  \cite{jaja86}.  Our main result of this section is Theorem \ref{rankG2W} where we prove the equality $\rank\ X\otimes_K\cW=6$ for a $2\times 2$ matrix $X$ of rank two.  \red{Theorem 5 has been independently obtained in \cite[version 2]{cjz17}.}  In Section \ref{sec:w2} we prove our main result: $\rank\ \cW\otimes\cW=8$, (Theorem \ref{thm:w2}).  Its proof follows from Proposition \ref{pp:2x3=6} which analyze the rank six decomposition of $X\otimes_K\cW$, where $\rank\ X=2$,  \red{and is based on the substitution method.}
We investigate the rank of \red{$\cW^{\ox n}$} in Section \ref{sec:wn}.  In Section \ref{sec:openprob} we list open problems related to our paper.
\section{Preliminary results}\label{sec:prelim}
Let $\bH$ be an $n$-dimensional Hilbert space with the inner product $\langle\bx,\by\rangle$ and the norm $\|\bx\|=\sqrt{\langle\bx,\bx\rangle}$.  Choose an orthonormal basis $\be_1,\ldots,\be_n$ in $\bH$.  Then $\bx=\sum_{i=1}^n x_i\be_i$ and we can identify $\bH$ with $\C^n$, where $\bx$ corresponds to $(x_1,\ldots,x_n)\trans\in\C^n$.
We denote $\bx=(x_1,\ldots,x_n)\trans$, and identify the inner product in $\bH$ with the standard inner product $\by^*\bx$ in $\C^n$, where $\by^*=(\bar y_1,\ldots,\bar y_n)$.
Let $\bH_i$ be Hilbert space of dimension $n_i$ for $i\in[d]$.  We identify  $\otimes_{i=1}^d\bH_i$ with $\otimes_{i=1}^d \C^{n_i}$.  Denoting the standard orthonormal basis of $\C^{n_i}$ by $\be_{1,i},\ldots,\be_{n_i,i}$, we obtain the  
elements of the standard basis of $\otimes_{i=1}^d\C^{n_i}$: $\otimes_{i=1}^d \be_{j_i,i}$, where $j_i\in[n_i]$ and $i\in[d]$.  Let $\cX\in\otimes_{i=1}^d\C^{n_i}$.
Then
\[\cX=\sum_{j_i\in[n_i],i\in[d]} x_{j_1,\ldots,j_d}\otimes_{i=1}^d \be_{j_i,i}.\]
Thus $\cX$ is represented by $d$-multiarray $[x_{j_1,\ldots,j_d}]$.  The space of the multiarrays is denoted by $\C^{\bn}=\C^{n_1\times\cdots\times n_d}$, where $\bn=(n_1,\ldots,n_d)$.  We will identify $\otimes_{i=1}^d\C^{n_i}$ with $\C^{\bn}$.  Assume that $n_1=\cdots=n_d=n$.  A tensor $\cX=[x_{j_1,\ldots,j_d}]\in\otimes^d\C^n$ is called \emph{symmetric} if $x_{j_1,\ldots,j_d}=x_{j_{\sigma(1)},\ldots,j_{\sigma(d)}}$ for any permutation $\sigma$ of the set $[d]$.  We denote by $S^{d,n}\subset \otimes^d\C^n$
the space of symmetric $d$-tensors on $\C^n$.  Recall that the spaces $\otimes^d\C^2\supset S^{d,2}$ are called the spaces of $d$-qubits and $d$-symmetric qubits respectively.

Let $GL(\C^n)$ be the general linear group acting on $\C^n$.  Denote $GL(\bn)=GL(\C^{n_1})\times\cdots\times GL(\C^{n_d})$.  Then $GL(\bn)$ acts on the space $\C^{\bn}$ as the following subgroup of $GL(\C^{N})$, where $N=n_1\cdots n_d$.  Namely,  $\otimes_{i=1}^d A_i= A_1\otimes\cdots\otimes A_d\in GL(\bn)$ acts on rank one tensor as follows: $(\otimes_{i=1}^d A_i)\otimes_{i=1}^d \bx_i=\otimes_{i=1}^d (A_i\bx_i)$.
Two tensors $\cX,\cY\in\C^{\bn}$ are called equivalent if $\cY=(\otimes_{i=1}^d A_i)\cX$.
If $n_1=\cdots=n_d$ and $A_i=A$ for $i\in [d]$ we denote $\otimes_{i=1}^d A_i$ by $\otimes^d A=A^{\otimes d}$.  Note that if $\cX\in S^{d,n}$ then $(\otimes^d A)\cX\in S^{d,n}$.

We first recall the well known characterization of the $\rank\ \cX$ for a $d$-tensor $\cX\in\C^{\bn}$.  See for example Proposition 2.1 in \cite{FRIEDLAND2012478} for the case $d=3$.
\begin{lemma}\label{charrnakten}
Let $\cX\in\otimes_{i=1}^d \C^{n_i}$.  Write $\cX=\sum_{j=1}^{n_1} \be_{j,1}\otimes\cX_j$, where $\cX_j\in \otimes_{i=2}^{d} \C^{n_i}$.  Denote by $\bW$
the subspace spanned by $\cX_1,\ldots,\cX_{n_1}$ in $\otimes_{i=2}^{d} \C^{n_i}$.
Then $\rank\ \cX$ is the dimension of a minimal subspace of $\otimes_{i=2}^{d} \C^{n_i}$ spanned by rank one tensors that contains $\bW$.  In particular, $\rank\ \cX\ge \dim\bW$.
\end{lemma}
Note that by changing the factor $\C^{n_i}$ with $\C^{n_1}$ we can apply Lemma \ref{charrnakten} to any $d-1$ factors: $(\otimes_{j=1}^{i-1}\C^{n_j})\otimes (\otimes_{j=i+1}^d \C^{n_j})$.

We  next recall Kruskal's theorem for $3$-tensors \cite{KRUSKAL197795}. Let $\bx_1,\ldots,\bx_p\in\C^q$.  Then the Kruskal rank of $\{\bx_1,\ldots,\bx_p\}$, denoted krank$(\bx_1,\ldots,\bx_p)$, is the maximal number $k$ such that any $k$ vectors in $\{\bx_1,\ldots,\bx_p\}$ are linearly independent.  
Assume that $\cX\in \C^{l}\otimes\C^m\otimes \C^n$, and we are given its decomposiiton in terms of rank one tensors:
\begin{equation}\label{Xdecomp}
\cX=\sum_{i=1}^r \bx_i\otimes \by_i\otimes \bz_i.
\end{equation}
Suppose that
\begin{equation}\label{Kruskalcon}
\mathrm{krank} (\bx_1,\ldots,\bx_r)+\mathrm{krank} (\by_1,\ldots,\by_r)+\mathrm{krank} (\bz_1,\ldots,\bz_r)\ge 2r+2.
\end{equation}
The $r=\rank\ \cX$ and the decomposiiton \eqref{Xdecomp} is unique.   That is, the rank one tensors $\bx_i\otimes \by_i\otimes \bz_i, i\in [r] $ are unique and linearly independent.  (One can change the order of the summation in \eqref{Kruskalcon}.)
It is possible to generalize Kruskal's theorem to $d$-partite tensors for $d>3$ by looking at these tensors as $3$-partitite tensors as in \cite{friedland16}.

We now state \emph{Strassen's direct sum conjecture} \cite{Strassen1973Vermeidung}.
Assume that $\cS\in\C^{\bm}, \cT\in\C^{\bn}$, where $\bm=(m_1,\ldots,m_d), \bn=(n_1,\ldots,n_d)$.  Then $\cS\oplus \cT$ is viewed as a tensor in $\C^{\bm+\bn}$. Clearly,
$
\rank\ (\cS\oplus \cT)\le\rank\ \cS + \rank\ \cT.
$
Strassen's direct sum conjecture states
\begin{equation}\label{StrDScon}
\rank\ (\cS\oplus \cT)=\rank\ \cS + \rank\ \cT.
\end{equation}
For $d=2$ (matrices) \eqref{StrDScon} holds. For $d=3$ equality holds if either $2\in\{m_1,m_2,m_3\}$ or $2\in\{n_1,n_2,n_3\}$, see \cite{jaja86}.

Otherwise, the conjecture is widely open.  Note that \eqref{StrDScon} fails for the border rank.  See A. Sch$\ddot{\text{o}}$nhage  counterexample in Example 4.5.2. of \cite{B2010Algebraic}.
Observe that $\cS\otimes_K \cT$ is an element of $\C^{\bm\circ\bn}$, where $\bm\circ \bn=(m_1n_1,\ldots,m_dn_d)$.

Denote by $\oplus^k\cT$ the direct sum of $k$-copies $\cT\in\C^{\bn}$.  Then \emph{restricted Strassen's} $k$-direct sum conjecture is  
\begin{equation}\label{StrDSconk}
\rank (\oplus^k \cT)=k\cdot\rank\cT.
\end{equation}
Clearly, if $ \rank\ (\oplus^\ell \cT)<\ell\rank\ \cT$ then $\rank\ (\oplus^k \cT)<k\cdot\rank\ \cT$ for each $k\ge\ell$. 

Denote
$\cG(k,d)=\sum_{i=1}^k \otimes^d \be_i\in \otimes^d\C^k$,
where $\be_1,\ldots,\be_k$ is the standard basis in $\C^k$. Clearly, $\rank\ \cG(k,d)=k$.
(For $d\ge 3$ one can use Kruskal's theorem by viewing $\cG(k,d)$ as a $3$-tensor as in \cite{friedland16}.)  Note that $\cG(2,3)$ is the GHZ state.  The following lemma is deduced straightforward.
\begin{lemma}\label{krestrlem}  Let $\cT\in\C^{\bn}$.  Then \eqref{StrDSconk} holds if and only if
\begin{equation}\label{krestrlem1}
\rank\ (\cG(k,d)\otimes_K \cT)= (\rank\ \cG(k,d))(\rank\ \cT)=k \cdot \rank\ \cT.
\end{equation}
Furthermore, if the above equality holds then 
\begin{equation}\label{krestrlem2}
\rank (\cG(k,d)\otimes \cT)=k\cdot\rank\cT.
\end{equation}
\end{lemma}
The result of JaJa-Takche \cite{jaja86} applied recursively to $(\oplus_{j=1}^p \cT_j)\oplus \cT_{p+1}$, \eqref{powrankin} and the above observations yield:
\begin{corollary}
	\label{JaTcor}  Assume that $\cT_1,\cdots,\cT_m\in \C^\bn$, where $\bn=(n_1,n_2,n_3)$.
Suppose that $2\in \{n_1,n_2,n_3\}$ and $\be_1,\ldots,\be_m$ is the standard basis in $\C^m$.  Then 
\begin{equation}
\rank\ \bigg(\sum^m_{j=1} (\otimes^3\be_j)\ox \cT_j\bigg)
=
\rank\ \bigg(\sum^m_{j=1} (\otimes^3\be_j)\ox_K \cT_j\bigg)
=
\sum^m_{j=1}\rank\ \cT_j.
\end{equation}
In particular, assume that $\cT\in\C^{\bn}$, where $\bn=(n_1,n_2,n_3)$ and $2\in\{n_1,n_2,n_3\}$.  Then
\begin{equation}\label{prodGk3id}
\rank\ (\cG(m,3)\otimes \cT)=\rank\ (\cG(m,3)\otimes_K \cT)=m\cdot\rank\ \cT.
\end{equation}
\end{corollary}
Recall Strassen's algorithm \cite{Strassen1969}, which states that  the product of $2\times 2$ matrices can be performed in $7$ mulitplications. It is known that the product of $2\times 2$ matrices can't be performed in $6$ multiplications \cite{hk71, Winograd1971On}. 
It is well known that the optimality of Strassen's algorithm
follows from the fact that the rank of the corresponding $3$-tensor $\cA=[a_{p,q,r}]\in \otimes^3\C^4$ is $7$.   The $64$ entriies of $a_{p,q,r}$ are either $0$ or $1$.  Furthermore there are $8$ entries which are equal to $1$.  It is easier to present $\cA$ using the Dirac notation \emph{bra-ket}.   View $\C^2\otimes \C^2$ as the space of $2\times 2$ matrices $\C^{2\times 2}$.  The standard basis in this space is $\be_i\otimes \be_j$, corresponding the matrices $\be_i\be_j\trans$ for $i,j\in[2]$.  In the bra-ket notation 
$\be_i\otimes \be_j$ corresponds to $|(i-1)(j-1)\rangle$.  To make transition to $\C^4$ we make the identification
\begin{equation}\label{conver2to1}
|00\rangle=|0\rangle,\quad |01\rangle=|1\rangle,\quad |10\rangle=|2\rangle, \quad |11\rangle=|3\rangle.
\end{equation}
Hence $|b\rangle$, where $b+1\in[4]$, represents an element in the basis $|st\rangle$, where $s,t\in\{0,1\}$. Thus, $a_{p,q,r}=1$ if and only if the product if $(\be_i\be_j\trans)(\be_{i'}\be_{j'}\trans)=\be_{\tilde i}\be_{\tilde j}\trans$, that is $i'=j, \tilde i=i, \tilde j=j'$.
Thus in bra-ket notation we have that $\cA=\sum_{i,j,k=0}^1|ij\rangle|jk\rangle |ik\rangle$.  By considering the isomorphism $|ik\rangle \mapsto |ki\rangle$  of $\C^2 \otimes \C^2$ we deduce that $\rank\ \cA=\rank\ \cB=7$, where
\bea
 \cB &=&
 \sum^1_{i,j,k=0} \ket{ij}\ket{jk}\ket{ki}
 \notag\\
 &=&
 \ket{000}+\ket{012}+\ket{120}+\ket{201}+\ket{321}+\ket{213}+\ket{132}+\ket{333}.
 \eea
\begin{lemma}\label{7rankex}  Let $\cT\in \C^2\otimes \C^2\otimes\C^4$ be the following tensor in Dirac's bra-ket notation $\cT=\ket{000}+\ket{012}+\ket{101}+\ket{113}$.   Then $\rank\ \cT=4$.  Furthermore
\begin{equation}\label{7rankeq}
\rank (\cG(2,2)\otimes_K \cT)=\rank \bigg( (\ket{00}+\ket{11})\ox_K \cT \bigg)=7.	
\end{equation}
\end{lemma}
\begin{proof}  Write $\cT=\sum_{i=0}^3 \cT_i\otimes \ket{i}$, where $\cT_i\in \C^2\otimes\C^2$.  As $\cT_{i-1}, i\in[4]$ is a rank one basis of $\C^2\otimes \C^2$, Lemma \ref{charrnakten} yields that $\rank\ \cT=4$.  Observe next that
\begin{eqnarray*}
(\ket{00}+\ket{11})\otimes_K \cT=(\ket{00}+\ket{11})\otimes_K(\ket{000}+\ket{012}+\ket{101}+\ket{113})=\\
\ket{000}+\ket{012}+\ket{101}+\ket{113}+\ket{220}+\ket{232}+\ket{321}+\ket{333}=\cC.
\end{eqnarray*}
Let $\psi:\C^4\to\C^4$ will be the isomorphism induced by
\[\phi(\ket{0})=\ket{0}, \quad \phi(\ket{1})=\ket{2}, \quad \phi(\ket{2})=\ket{1}, \quad \phi(\ket{4})=\ket{4}.\]
Denote by $\tilde \psi: \otimes^3\C^4$ the isomorphism induced by $\tilde\phi(\bx\otimes\by\otimes\bz)=\phi(\bx)\otimes\by\otimes\bz$.  Observe that $\tilde \psi$ preserves the rank of tensors in $\otimes^3\C^4$. Clearly, $\tilde \psi(\cC)=\cB$. 
So \eqref{7rankeq} follows from the fact that $\rank\ \cB=7$.
\end{proof}

\begin{theorem}\label{rankG2W}  Let $\cW\in \otimes^3\C^2$ be the state given by Dirac's notation $\ket{001}+\ket{010}+\ket{100}$.  Then
\begin{equation}\label{rankG2W1}
\rank\ \cG(k,2)\otimes_K \cW=\rank\ \cG(k,2)\otimes \cW=3k.
\end{equation}
\end{theorem}
\begin{proof}  Let $\cX\in \C^{2k}\otimes\C^{2k}\otimes \C^2$.  Write $\cX=X_1\otimes |0\rangle+X_2\otimes |1\rangle$, where $X_1,X_2\in\C^{(2k)\times (2k)}$.
The fundamental result of \cite{Ja1978Optimal} yields that the rank of $\cX$ can be determined completely by the Kronecker canonical form of the pair $(X_1,X_2)$.  In particular, $\rank\ \cX\le 3k$.
Assume that the span of $X_1,X_2$ is two dimensional with a basis $A_0,A_1\in\C^{4\times 4}$. \red{(Here by $\C^{4\times 4}=\C^4\otimes \C^4$ we denote the space for $4\times 4$ complex matrices.)}  Suppose furthermore that $A_0$ is invertible.  Then $\rank\ \cX=3k$ if and only if
the Jordan canonical form of $A_0^{-1}A_1$ consists of $k$ identical $2 \times 2$ Jordan blocks. 

Let $\cX=\cG(k,2)\otimes_K \cW$.  Clearly, 
\[\rank\ \cG(k,2)\otimes_K\cW\le \rank\ \cG(k,2)\otimes\cW\le(\rank\ \cG(k,2))(\rank\ \cW)=3k.\]
Hence it is enough to show that $\rank\ \cX=3k$.  
Note that $\cX\in \C^{2k}\otimes\C^{2k}\otimes \C^2$.  
   Observe next we can identify $\cG(k,2)$ with the $k\times k$ identity matrix $I_k$.
  \begin{eqnarray*} 
  \cG(2,2)\otimes_K \cW=I_k\otimes_K (|001\rangle+|010\rangle+|100\rangle)=
  (I_k\otimes_K(\ket{01}+\ket{10})\otimes \ket{0} +(I_k\otimes_K\ket{00})\otimes\ket{1}=
  A_0\otimes|0\rangle+A_1|1\rangle.
  \end{eqnarray*}
  Here
  \[A_0=I_k\otimes_K B_0, \quad A_1=I_k\otimes_K B_1, \quad B_0=\left[\begin{array}{cc}0&1\\1&0\end{array}\right], \quad 
  B_1=\left[\begin{array}{cc}1&0\\0&0\end{array}\right].
  \]
  Note that $B_0^2=I_2$ and $B_0B_1=J=\left[\begin{array}{cc}0&0\\1&0\end{array}\right]$.   Hence $A_0^{-1}=I_k\otimes_K B_0$ and
  $A_0A_1=I_k\otimes_K J$.  Thus the Jordan canonical form of $A_0A_1$ consists of $k$ idientical Jordan blocks $J$.    Hence $\rank\ \cX=3k$.
\end{proof}

\section{The rank of $\cW\otimes\cW$}
\label{sec:w2}
Let $\phi:\C^{\bm}\to \C$ be a nonzero linear transformation.  Then $\phi$ extends to two  linear transformations $\tilde\phi:\C^{\bm}\times \C^{\bn}\to \C^{\bn}, \hat\phi:\C^{\bn}\times \C^{\bm}\to \C^{\bn}$ as follows:  
\[\tilde\phi(\cX\otimes\cY)=\phi(\cX)\cY, \quad \hat\phi(\cY\otimes\cX)=\phi(\cX)\cY, \quad\textrm{for all }\cX\in\C^{\bm},\cY\in\C^{\bn}.\]
\begin{lemma}
\label{le:prod}
Let $\cX\in \C^{\bm}$ be a rank one $d$-tensor, and $\cY\in\C^{\bn}$ be a $d'$-partitie state. Then $\rank\ (\cX\otimes \cY)=\rank\ \cY=r$. 
Assume furthermore that $\cX\otimes\cY=\sum_{i=1}^r \cZ_i$, where $\cZ_i\in\C^{\bm}\otimes \C^{\bn}$ is a rank one tensor. Then $\cZ_i=\cX\otimes \cY_i$, where $\cY_i\in \C^{\bn}$ is a rank one tensor, and $\sum_{i=1}^r \cY_i=\cY$.	
\end{lemma}
\begin{proof}  
Using induction, it is enough to show the case where $d=1$, i.e. $\bm=(m)$, $m>1$ and $\cX=\bx$ is a nonzero vector in $\C^m$.  Assume that $r=\rank\ \cY,r'=\rank\ \bx\otimes \cY$.  Clearly, $r'\le r$.  Then  $\bx\otimes\cY=\sum_{i=1}^{r'} \bx_i\otimes \cY_i$, where $\bx_i$ and $\cY_i$ are rank one tensors.  Let $\phi:\C^{m}\to \C$ be a linear functional such that $\phi(\bx)=1$.  Thus 
\[\tilde\phi(\bx\otimes\cY)=\phi(\bx)\cY=\cY=\sum_{i=1}^{r'}\phi(\bx_i)\cY_i.\]
As $\rank\ \cY=r$ it follows that $r=r'$.  Furthermore, $\phi(\bx_i)\ne 0$ for $i\in[r]$.
Assume that $\bx_i$ is not proportional to $\bx$.  Then there exists $\phi$ such that $\phi(\bx)=1$ and $\phi(\bx_{i})=0$.  This will contradict the assumption that $\rank\ \cY=r$.  Without loss of generality we can assume that $\bx_i=\bx$ for each $i\in[m]$.
\end{proof}

\begin{lemma}
\label{le:a11}
Let $\cT=\ba_{1,1}\ox \cdots \ox \ba_{d,1}+\ba_{1,2}\ox \cdots \ox \ba_{d,2}
=
\cBb_{1,1}\ox \cdots \ox \cBb_{d,1}+\cBb_{1,2}\ox \cdots \ox \cBb_{d,2}
$ be a nonzero $d$-tensor, where $d\ge 3$.

(i) If $\cT$ has rank two 
then $\ba_{j,1}\propto \ba_{j,2}$ if and only if $\cBb_{j,1}\propto \cBb_{j,2}$ if and only if $\ba_{j,1}\propto \ba_{j,2}\propto \cBb_{j,1}\propto \cBb_{j,2}$. 

(ii) If $\cT$ has rank two, $\ba_{j,1}$ and $\ba_{j,2}$ are linearly independent for $j\in S\subseteq [d]$ and $\abs{S}>2$, then $\textrm{span}\{\ba_{1,1}\ox \cdots \ox \ba_{d,1},\ba_{2,2}\ox \cdots \ox \ba_{d,2}\}=\textrm{span}\{
\cBb_{1,1}\ox \cdots \ox \cBb_{d,1},\cBb_{1,2}\ox \cdots \ox \cBb_{d,2}\}$.

(iii) If $\cT$ has rank one, $\ba_{j,i},\cBb_{j,i}$ are all nonzero, then $\ba_{j,1}\propto \ba_{j,2}$ and $\cBb_{k,1}\propto \cBb_{k,2}$  for $j\in S\subset [d]$, $k\in T\subset [d]$, and $\ba_{j,1}\propto \ba_{j,2}\propto \cBb_{j,1}\propto \cBb_{j,2}$ for $j\in S\cap T$ and $\abs{S}=\abs{T}=d-1$. 
\end{lemma}
\begin{proof} (i)  Assume that $\rank\ \cX=2$ and $\ba_{j,1}\propto \ba_{j,2}$.  By permuting the factors of the tensor products, we can assume that $j=1$.  Then Lemma \ref{le:prod} yields that $\cBb_{1,1}$ and $\cBb_{1,2}$ are nonzero vectors proportional to
$\ba_{1,1}$ and $\ba_{1,2}$.

(ii) Without loss fo generality we can assume that $\ba_{j,1}$ and $\ba_{j,2}$ are linearly independent for $j=1,2,3$.  View $\cT$ as a $3$-tensor $\cT'$ on the tensor product $\C^{n_1}\otimes \C^{n_2}\otimes \C^m$, where $ \C^m=\otimes_{i=3}^d \C^{n_i}$. Set
$\ba_{3,j}'=\otimes_{i=3}^d \ba_{i,j},  \cBb_{3,j}'=\otimes_{i=3}^d \cBb_{i,j}$, for $j=1,2$.
As $\rank\ \cT=2$, and the pairs $\ba_{3,1},\ba_{3,2}$ and $\cBb_{3,1},\cBb_{3,2}$  are linearly independent it follows that the pairs $\ba_{3,1}',\ba_{3,2}'$ and $\cBb_{3,1}',\cBb_{3,2}'$ are linearly independent.  Use Kruskal's theorem for $\cT'$ to deduce that $\{\ba_{1,1}\ox \cdots \ox \ba_{d,1},\ba_{2,2}\ox \cdots \ox \ba_{d,2}\}=\{
\cBb_{1,1}\ox \cdots \ox \cBb_{d,1},\cBb_{1,2}\ox \cdots \ox \cBb_{d,2}\}$.

(iii) We first discuss the equality $\cT=\ba_{1,1}\ox \cdots \ox \ba_{d,1}+\ba_{1,2}\ox \cdots \ox \ba_{d,2}$, the assumption that $\rank\ \cT=1$ and all $\ba_{i,1},\ba_{i,2}$ are nonzero.  Use Kruskal's theorem, as in the proof of part (ii), to deduce  that at most two pairs of vectors $\ba_{j,1}, \ba_{j,2}$ are linearly independent.  Without loss of generality we can assume that $\ba_{j,1}=\ba_{j,2}$ for $j>3$.  Use Lemma  
\ref{le:prod} to deduce that $1=\rank\ \cT$ is the rank of the matrix $\ba_{1,1}\otimes \ba_{2,1}+\ba_{1,2}\otimes \ba_{2,2}$.  Clearly this matrix has rank one iff and only $\ba_{1,1}\otimes \ba_{2,1}\ne -\ba_{1,2}\otimes \ba_{2,2}$, and 
either $\ba_{1,1}\propto \ba_{1,2}$ or $\ba_{2,1}\propto \ba_{2,2}$.  In particular, there exists $S\subset [d], \abs{S}=d-1$ such that $\ba_{i,1}\propto \ba_{i,2}$ for $i\in S$.
Similarly, there exists $T\subset [d], \abs{T}=d-1$ such that $\cBb_{i,1}\propto \cBb_{i,2}$ for $i\in T$.  Hence $\ba_{j,1}\propto \ba_{j,2}\propto \cBb_{j,1}\propto \cBb_{j,2}$ for $j\in S\cap T$.
\end{proof}

\begin{lemma}
\label{le:symqubit} 
(i) If two symmetric $d$-qubit tensors $\cX$ and $\cY$ are equivalent, then there exists $A\in GL(\C^2)$ such  $A^{\ox d}\cX=\cY$.
\\
(ii) Suppose $\ba,\cBb,\cBc,\cBd$ are pairwise linearly independent vectors in $\cBC^2$. If $\ba\ox \ba+\cBb\ox \cBb=\cBc\ox\cBc+\cBd\ox \cBd$ then $\cBc=\a \ba+ \b \cBb$ and $\cBd=\pm(\b \ba- \a \cBb)$, where $\a,\b$ are nonzero complex numbers such that $\a^2+\b^2=1$.
\\
(iii) Suppose $\ba,\cBb \in\cBC^2$, and $x,y$ are two complex numbers. Then the 3-tensor $\cZ:=x\ba^{\ox3}+y\cBb^{\ox3}+(\ba+\cBb)^{\ox3}$ is equivalent to the tensor \red{$\cW$} if and only if $\ba$ and $\cBb$ are linearly independent, $xy\ne0$ and $4xy=(x+y+xy)^2$. 
\end{lemma}
\begin{proof}
(i) is proved in \cite{mkg2010}. 

(ii) Let $A\in GL(\C^2)$ such that $A\ba=\cBc, A\cBb=\cBd$.  Then $A$ is a complex orthogonal matrix.  

(iii) The result \cite{dvc2000} yields that $\cZ$ is equivalent to $W$ state if and only if $\rank\ \cZ=3$.  If $\cZ$ has rank $3$ then $\ba$ and $\cBb$ are linearly independent and $xy\ne 0$. Assume that this is the case.  Let $A\in GL(\C^2)$ such that $A\ba=\ket{0}, A\cBb=\ket{1}$.  
Then $\cZ'=A^{\otimes 3}\cZ=x\ket{0}^{\otimes 3}+y\ket{1}^{\otimes 3}+(\ket{0}+\ket{1})^{\otimes 3}$.  Write $\cZ'$ as $Z_1\otimes \ket{0}+Z_2\otimes \ket{1}$, where
\[Z_1=\left[\begin{array}{cc}x+1&1\\1&1\end{array}\right], \quad Z_2=\left[\begin{array}{cc}1&1\\1&1+y\end{array}\right].\]
Observe next that $Z_1-Z_2$ is a diagonal invertible matrix.  A well known result, e.g. \cite{Ja1978Optimal}, claims that $\rank\ \cZ= \rank\ \cZ'=3$ if and only if
the Jordan canonical form of $Z:=(Z_1-Z_2)^{-1}(Z_1+Z_2)$ has one Jordan block.
That is $Z$ has a double eigenvalue and $Z$ is not diagonizable.  The assumption that 
$Z$ has a double eigenvalue is equivalent to the condition $4xy=(x+y+xy)^2$.  If $Z$ was diagonazible then 
$Z=\lambda I_2$, where $I_2$ is the identity matrix.  As $Z$ has nonzero off-diagonal entries it follows that the Jordan canonical form of $Z$ has one Jordan block if $xy\ne 0$ and  $4xy=(x+y+xy)^2$.
Hence $\cZ$ has rank $3$ if and only if $\ba$ and $\cBb$ are linearly independent, $xy\ne0$ and $4xy=(x+y+xy)^2$.  
\end{proof}

\begin{proposition}
\label{pp:2x3=6} 
Let $X\in\C^{2\times 2}$ and $\cY\in \otimes^3\C^{2}$, where $\rank\ X=2$ and $\rank\ \cY=3$.  Then
\begin{enumerate}
\item $\rank(X\ox \cY)=6$.
\item Assume that
\begin{equation}
\label{eq:xy}
X\ox \cY=\sum^6_{j=1}\otimes_{i=1}^5\cBc_{j,i}.
\end{equation}
Then $\{\cBc_{1,1}\ox \cBc_{1,2}\ox\cBc_{1,i},\ldots,\cBc_{6,1}\ox \cBc_{6,2}\ox\cBc_{6,i}\}$ are linearly dependent for $i=3,4,5$, and
$\{\cBc_{1,p}\ox \cBc_{1,q}\ox\cBc_{1,r},\ldots,\cBc_{6,p}\ox \cBc_{6,q}\ox\cBc_{6,r}\}$ are linearly independent for $p=1,2$ and $3\le q<r \le 5$.
\end{enumerate}
\end{proposition}
\begin{proof}
\emph{1} Clearly, $X$ is equivalent to $\cG(2,2)$ and it is well known that $\cY$ is equivalent to $\cW$.  Hence $X\otimes \cY$ is equivalent to $\cG(2,2)\otimes \cW$.  Theorem \eqref{rankG2W} yields that $\rank\ \cG(2,2)\otimes \cW=6$.   Hence $\rank(X\ox \cY)=6$.

\emph{2} Without loss of generality we may assume that $X=\cG(2,2)$ and $\cY=\cW$.
Let $\phi:\C^2\to \C$ be a nonzero linear functional.  A straightforward calculation shows that $\tilde\phi(\cG(2,2))\ne 0$ and $\tilde\phi(\cW)\ne 0$.  Furthermore, $\rank\ \tilde\phi(\cW)=1$ if and only if $\phi(\ket{0})=0$.

Next we observe that the set $\{\cBc_{1,i}\ldots,\cBc_{6,i}\}$ contain two independent vectors for each $i\in[5]$.  Suppose to the contrary that $\cBc_{1,1}\propto\cdots\propto\cBc_{6,1}$.
Then there exists a nonzero linear functional $\phi:\C^2\to\C$ such that $\phi(\cBc_{j,1})=0$ for $j\in[6]$.  This would imply that $\tilde\phi(\cG(2,2))\otimes \cW=0$.  Hence $\tilde\phi(\cG(2,2))=0$, which is impossible.  Therefore $\{\cBc_{1,1}\ldots,\cBc_{6,1}\}$ contains two independent vectors. View $X\otimes \cY$ as a $5$-tensor on $\otimes_{i=1}^5 \bU_i$, where each $\bU_i=\C^2$.  Iinterchange the two factors $U_1$ and $U_2$ to deduce that 
$\{\cBc_{1,2}\ldots,\cBc_{6,2}\}$ contains two independent vectors.  
By considering $\cY\otimes \cG(2,2)$ and using the fact that $\tilde\phi(\cW)\ne 0$ for any nonzero functional $\phi$ we deduce that $\{\cBc_{1,3}\ldots,\cBc_{6,3}\}$ contain two independent vectors. By interchanging $\bU_3$ with $\bU_i$ for $i=4,5$ we deduce that $\{\cBc_{1,i}\ldots,\cBc_{6,i}\}$ contain two independent vectors for $i=4,5$.
 Thus we showed that the set $\{\cBc_{1,i}\ldots,\cBc_{6,i}\}$ contains two independent vectors for each $i\in[5]$. 

Next we observe that the set $\{\cBc_{1,4}\otimes\cBc_{1,5},\ldots,\cBc_{6,4}\otimes\cBc_{6,5}\}$ contains $3$-linearly independent rank one matrices.  Assume to the contrary that there are two linearly independent rank one matrices $\bu\otimes\bv$ and $\bx\otimes \by$ in $\C^2\otimes\C^2$ whose span contains $\{\cBc_{1,4}\otimes\cBc_{1,5},\ldots,\cBc_{6,4}\otimes\cBc_{6,5}\}$. Let
$\psi:\C^2\otimes\C^2\to \C$ be a nonzero linear functional such that $\psi(\bu\otimes\bv)=\psi(\bx\otimes\by)=0$.  Hence $\tilde\psi(\cW)\otimes \cG(2,2)=0$, which implies that $\tilde\psi(\cW)=0$.  This condition is equivalent to $\psi(\ket{01}+\ket{10})=\psi(\ket{00})=0$.  Since $\psi$ was any nonzero linear functiona that vanishes on $\bu\otimes\bv$ and $\bx\otimes \by$, it follows that the two matrices $\ket{01}+\ket{10}$ and $\ket{00}$ are linear combinations of $\bu\otimes\bv$ and $\bx\otimes \by$.  Lemma \ref{charrnakten} yields that $\rank\ \cW\le 2$ which is false.  Hence $\{\cBc_{1,4}\otimes\cBc_{1,5},\ldots,\cBc_{6,4}\otimes\cBc_{6,5}\}$
 contain three linearly independent rank one matrices. By permuting accordingly the factors $\bU_3$, $\bU_4$ and $\bU_5$ we deduce that $\{\cBc_{1,q}\otimes\cBc_{1,r},\ldots,\cBc_{6,q}\otimes\cBc_{6,r}\}$ contain three linearly independent rank one matrices for $p\ne r$ and $p,r\in\{4,3,5\}$.

We now show that the six $3$-tensors $ \{\cBc_{1,1}\ox \cBc_{1,2}\ox\cBc_{1,3},\ldots,\cBc_{6,1}\ox \cBc_{6,2}\ox\cBc_{6,3}\}$ are linearly dependent.
 View the tensor $\cG(2,2)\otimes \cW$ as an $8\times 4$ matrix,denoted by $Z$, by grouping the first three factors and the last $2$ factors: $(\bU_1\otimes \bU_2\otimes \bU_3)\otimes (\bU_4\otimes \bU_5)$.
 Then
 \[\cG(2,2)\otimes \cW=\cG(2,2)\otimes (\ket{001} +\ket{010}+\ket{100})=(\cG(2,2)\otimes\ket{0})\otimes(\ket{01}+\ket{10})++
 (\cG(2,2)\otimes\ket{1})\otimes\ket{00}.\]
 Hence $\rank\ Z=2$.
On the other hand \eqref{eq:xy} states $Z=\sum^6_{j=1}(\cBc_{j,1}\ox \cBc_{j,2}\ox \cBc_{j,3})\ox (\cBc_{j,4}\ox \cBc_{j,5})$.  
Assume to the contrary that the six $3$-tensors  $\cBc_{j,1}\ox \cBc_{j.2}\ox \cBc_{j,3},j=1,\ldots,6$ are linearly independent.  Then the span of the row space of $Z$ is the span of the six  matrices $ \cBc_{j,4}\ox \cBc_{j,5},j=1,\ldots,6$, which is at least three dimensional.  Hence $\rank\ Z\ge 3$, which contradicts the previous equality $\rank\ Z= 2$.
 
We now show that the six $3$-tensors $\{\cBc_{1,2}\ox \cBc_{1,4}\ox\cBc_{1,5},\ldots,\cBc_{6,2}\ox \cBc_{6,4}\ox\cBc_{6,5}\}$ are linearly independent.
 Let $\cT=[t_{p,q,r}]\in \C^4\otimes \C^4\otimes \C^2$ be given by $\cT=\sum_{j=1}^6 (\cBc_{j,1}\otimes \cBc_{j,3})\otimes (\cBc_{j,2}\otimes \cBc_{j,4})\otimes \cBc_{j,5}$.
Theorem \ref{rankG2W} yields that $\rank\ \cT=6$.  Let $T_{p}=[t_{p,q,r}]_{q=r=1}^{4,2}\in\C^{4\times 2}, p\in[4]$ be the four frontal sections of $\cT$.  Lemma \ref{charrnakten} yields that the rank of $\cT$ is dimension of the minimal subspace in $\C^{4\times 2}$ spanned by rank one matrices that contains $T_1,\ldots,T_4$.  Clearly, the subspace spanned by six rank one matrices $(\cBc_{j,2}\otimes \cBc_{j,4})\otimes \cBc_{j,5}, j\in [6]$ contains $T_1,\ldots,T_4$.  Since $\rank\ \cT=6$ it follows these six rank one matrices are linearly independent. 
By permuting the factors $\bU_1,\bU_2$ and $\bU_3,\bU_4,\bU_5$ we deduce that
$\{\cBc_{1,p}\ox \cBc_{1,q}\ox\cBc_{1,r},\ldots,\cBc_{6,p}\ox \cBc_{6,q}\ox\cBc_{6,r}\}$ are linearly independent for $p=1,2$ and $3\le q<r \le 5$.
This completes the proof of the proposition.	
\end{proof}

\begin{theorem}
\label{thm:w2}
The 6-tensor $\cW\otimes \cW\in \bU_1\ox \bU_2\ox \bU_3\ox \bU_4\ox \bU_5\ox \bU_6$, where $\bU_i=\C^2$ for $i\in[6]$,  has rank eight.
\end{theorem}
\begin{proof}
Consider the tensor $\cW\otimes_K\cW$,  which is a 3-vector in $(\bU_1\ox \bU_4)\ox(\bU_2\ox \bU_5)\ox(\bU_3\ox \bU_6)$.  Then $\rank \cW\otimes_K\cW=7$  \cite{ycg10,cds10}.  
So the $6$-vector $W^{\ox 2}$
has rank at least seven. On the other hand $\cW\otimes\cW$ has rank at most eight \cite{cjz17}. So the assertion holds if we can disprove that $\cW\otimes\cW$  \red{has} rank seven. 

Assume to the contrary that 
\begin{equation}
\label{W2rank7}	
\cW\otimes \cW=\sum^7_{j=1}\otimes_{i=1}^6 \ba_{j,i}.
\end{equation}
We claim that for each $j\in[7]$ either $\ba_{j,1}\propto\ba_{j,2}\propto\ba_{j,3}\propto\ket{0}$ or
$\ba_{j,4}\propto\ba_{j,5}\propto\ba_{j,6}\propto\ket{0}$.  Assume that this claim does not hold.  Then by rearranging the seven summands in \eqref{W2rank7}	and permuting the factors $\bU_1,\bU_2,\bU_3$ and $\bU_4,\bU_5,\bU_6$ we can assume that neither $\ba_{7,1}$ nor $\ba_{7,4}$ are proportional to $\ket{0}$.  Let $\phi_i:\C^2\to \C$ be a nonzero linear functional such that $\phi_i(\ba_{7,i})=0$ for $i\in [6]$.  Hence $\tilde\phi_1(\cW)$ and $\tilde\phi_4(\cW)$ are rank two matrices. The equality \eqref{W2rank7} yields
\[\tilde\phi_1(\cW)\otimes \cW=\sum^6_{j=1}\phi_1(\ba_{j,1})\otimes_{i=2}^6 \ba_{j,i}.\]
Since $\tilde\phi_1(W)$ has rank two, Proposition \ref{pp:2x3=6} yields the following facts.  First,  $\rank\ \tilde\phi_1(\cW)\otimes \cW=6$. Hence $\phi_1(\ba_{j,1})\ne 0$ for $j\in[6]$.  \red{Second, the six rank one tensors
$\ba_{j,3}\otimes \ba_{j,5}\otimes \ba_{j,6}$ are linearly independent for $j\in [6]$.  (This choice correspond to the choice $p=2, q=4, r=5$ in Proposition \ref{pp:2x3=6}.)} 
 Swap the factors $\bU_i$ with $\bU_{i+3}$ for $i\in [3]$ in \eqref{W2rank7} to deduce 
\[\cW\otimes \cW=\sum^7_{j=1}(\otimes_{i=4}^6 \ba_{j,i})\otimes (\otimes_{i=1}^3 \ba_{j,i}).\] 
Therefore
\[\tilde\phi_4(\cW)\otimes \cW=\sum^7_{j=1}(\phi_4(\ba_{j,4})\otimes_{i=5}^6 \ba_{j,i})\otimes (\otimes_{i=1}^3 \ba_{j,i}).\] 
As $\tilde\phi_4(W)$ has rank two, Proposition \ref{pp:2x3=6} yields the following facts.
First, $\phi_4(\ba_{j,4})\ne 0$ for $j\in[6]$.  Second the six rank one tensors $\ba_{j,5}\otimes\ba_{j,6}\otimes \ba_{j,3},j\in[6]$ are linearly dependent. \red{(This choice corresponds the choice $i=5$ in Proposition \ref{pp:2x3=6}.)  Therefore} the six rank one tensors $\ba_{j,3}\otimes\ba_{j,5}\otimes\ba_{j,6} ,j\in[6]$ are linearly dependent.  This contradicts the previous claim that  the six rank one tensors
\red{$\ba_{j,3}\otimes \ba_{j,5}\otimes \ba_{j,6}$ are linearly independent for $j\in [6]$.}    Hence for each $j\in[7]$ either $\ba_{j,1}\propto\ba_{j,2}\propto\ba_{j,3}\propto\ket{0}$ or
$\ba_{j,4}\propto\ba_{j,5}\propto\ba_{j,6}\propto\ket{0}$. 

Clearly, we can't have that $\ba_{j,1}\propto\ba_{j,2}\propto\ba_{j,3}\propto\ket{0}$ for $j\in[7]$.  Otherwise $\cW\otimes \cW=(\otimes^3\ket{0})\otimes \cA$ for some $\cA\in\otimes^3\C^2$.  Lemma \ref{le:prod} yields that $\rank\ ((\otimes^3\ket{0})\otimes \cA)=\rank\ \cA\le 3 $ which contradicts the inequality $\rank\ \cW\otimes \cW \ge 7$.  Similarly, we can't have that $\ba_{j,4}\propto\ba_{j,5}\propto\ba_{j,6}\propto\ket{0}$ for $j\in[7]$.   Therefore $\cW\otimes \cW=(\otimes^3\ket{0})\otimes \cA+\cB\otimes (\otimes^3\ket{0})$, where $\cB\in\otimes^3\C^2$.  Hence 
\[\rank\ \cW\otimes \cW\le \rank\ ((\otimes^3\ket{0})\otimes \cA)+ \rank\ (\cB\otimes (\otimes^3\ket{0})\le 6.\]
This contradict the inequality $\rank\ \cW\otimes \cW \ge 7$. Hence $\rank\ \cW\otimes \cW=8$.
\end{proof}

\section{Estimating the rank of $\cW^{\ox n}$}
\label{sec:wn}

In this section we estimate the rank of $\cW^{\ox n}$ with $n>2$. 
It was shown by Zuiddam \cite{Zuiddam2017A} that $\rank\ \otimes_K^3\cW=16$.  Hence $\rank\ \cW^{\otimes 3}\ge 16$.
It has been mentioned in \cite[Remark 14]{cjz17} that $\rank\ \cW^{\ox 3}\le21$. 
The following theorem improves on the above upper bound:
\begin{theorem}
\label{ubrankw3}
\begin{equation}\label{brnkw3}
16\le \rank\ \cW^{\otimes 3}\le 20.
\end{equation}
\end{theorem}
We first recall well known characterization \cite{Ja1978Optimal}: 
\begin{lemma}\label{rank3tens} Let $\cT\in\otimes^3\C^2$.  Assume that $\be_1=(1,0)\trans=\ket{0},\be_2=(0,1)\trans=\ket{1}$ is a standard basis in $\C^2$.  Set $\cT= \be_1\otimes T_1+\be_2\otimes T_2$, where $T_1,T_2\in\C^{2\times 2}$.  Then $\rank\ \cT=3$ if and only if span$(T_1,T_2)$ is two dimensional and spanned by $A,B\in\C^{2\times 2}$, where $A$ is invertible and $A^{-1}B$ is a nondiagonizable matrix.
\end{lemma} 
\begin{lemma}\label{rank1Wpert}  Let $\cB\in\otimes^3\C^2$ be a rank one tensor.
Consider the one parameter family of tensors $\cW+t\cB$ for $t\in\C$.  Then  for a random choice of $\cB$ the rank of $\cW+t\cB$  is two unless $t\in\{0,t_1\}$.  (For $t=0$ and $t=t_1\ne 0$ the rank of  $\cW+t\cB$ is three.)  In particular
\begin{enumerate}
\item For $\cB=\be_1\otimes\be_2\otimes \be_2$ the rank of $\cW+t\cB$ is two for $t\ne 0$.
\item For $\cB=\be_1\otimes\be_1\otimes\be_1$ the rank of $\cW+t\cB$ is three for
all $t\in\C$.
\item For $\cB\in\{\be_1\otimes\be_1\otimes\be_2,\be_1\otimes\be_2\otimes\be_1\}$ the rank of $\cW+t\cB$ is three for all $t\ne -1$.  The rank of $\cW-\cB$ is two.
\item For $\cB=\be_2\otimes \bx\otimes \by, \bx=(x_1,x_2)\trans\ne 0, \by=(y_1,y_2)\trans$, the tensor $\cW+t\cB$ has rank two for $t\not\in\{0, t_1\}$, except in the following cases: 
\begin{enumerate} 
\item If $x_1x_2y_1y_2\ne 0$ and $x_2y_1+x_1y_2=0$ then $\cW+t\cB$ has rank two for $t\ne 0$.
\item If $x_1=y_1y_2=0$ then $\cW+t\cB$ has rank two for $t\ne 0$.
\item $x_2=0$.  If $y_2\ne 0$ then $\cW+t\cB$ has rank two for $t\ne 0$.  If $y_2=0$ then $\cW+t\cB$ has rank three for $t\ne t_1$.
\item If $y_1=x_1x_2=0$ then $\cW+t\cB$ has rank two for $t\ne 0$.
\item $y_2=0$. If $x_2\ne 0$ then $\cW+t\cB$ has rank two for $t\ne 0$.  If $x_2=0$ then $\cW+t\cB$ has rank three for $t\ne t_1$.
\end{enumerate} 
\end{enumerate}
\end{lemma}
\begin{proof}  Let
\[\cW=\be_1\otimes W_1+\be_{2}\otimes W_2, \quad W_1=\left[\begin{array}{cc}0&1\\1&0\end{array}\right], W_2=\left[\begin{array}{cc}1&0\\0&0\end{array}\right].\]
Assume that $\cB=\bu\otimes \bx\otimes \by$ where $\bu=(u_1,u_2)\trans$ and $\bx,\by$ as above.  So
$\cW+t\cB=\be_1\otimes W_1(t)+\be_2\otimes W_2(t)$, where $W_1(t)=W_1+tu_1\bx\by\trans, W_2(t)=W_2+tu_2\bx\by\trans$.  Assume that $u_1u_2\ne 0$.  Then 
\[W_3:=u_2W_1-u_1W_2=u_2W_1(t)-u_1W_2(t)=\left[\begin{array}{cc}-u_1&u_2\\u_2&0\end{array}\right], W_3^{-1}=u_2^{-2}\left[\begin{array}{cc}0&u_2\\u_2&u_1\end{array}\right].\]
Assume that $t\ne 0$ and let $s=\frac{1}{t}$.  Define 
\[W_4(s)=\frac{u_2^2}{t}W_3^{-1}W_2(t)=\left[\begin{array}{cc}x_1'y_1&x_1'y_2\\su_2+x_2'y_1&x_2'y_2\end{array}\right], (x_1',x_2')\trans=u_2^2W_3^{-1}\bx.\]
The condition that $W_4(s)$ has a double eigenvalue is $($trace $W_4(s))^2=4$det $W_4$. For $\cB$ chosen at random, this will give a linear equaition in $s$ whose solution is $s\ne 0$.  So $t_1=\frac{1}{s}$.  As for random $\cB$ $x_1'y_2\ne 0$ it follows that $\cW+t_1\cB$ has rank $3$ and for $t\not\in\{0,t_1\}$ the tensor $\cW+t\cB$ has rank two.  Other claims of the lemma follow straightforward using
Lemma \ref{rank3tens} and the above arguments.
\end{proof}

\begin{corollary}\label{srankonepertW}
Assume that $\cB\in\otimes^3\C^2$ is a rank one tensor proportional to one of the tensors $\be_2\otimes\be_2\otimes\be_2, \be_2\otimes\be_2\otimes\be_1,\be_2\otimes\be_1\otimes\be_2, \be_1\otimes\be_2\otimes\be_2$.  Then $\cW+t\cB$ has rank two for $t\ne 0$.
\end{corollary}

\textbf{Proof of Theorem \ref{ubrankw3}}.\begin{proof}
Let $\cB_1,\cB_2,\cB_3$ be rank one $3$-tensors of the form given by Corollary \ref{srankonepertW}.  Set $\cX_i=\cW-\cB_i$ for $i\in[3]$.  Then $\rank\ \cX_i=2$ for $i\in[3]$.  Observe next  that
\begin{eqnarray}
\label{eq:w32}
\cW^{\ox 3}=\otimes_{i=1}^3(\cX_i+\cB_i)=
&=&
\cX_1
\ox
(\cX_2+y\cB_2)
\ox	
(\cX_3+{1\over y}\cB_3)
\notag\\
&+&
\cB_1
\ox
(\cX_2+z\cB_2)
\ox	
(\cX_3+{1\over z}\cB_3)
\notag\\
&+&
\bigg(
(1-{1\over y})\cX_1+
(1-{1\over z})\cB_1
\bigg)\ox \cX_2 \ox \cB_3
\notag\\
&+&
\bigg(
(1-y)\cX_1+
(1-z)\cB_1
\bigg)\ox \cB_2 \ox \cX_3,
\end{eqnarray}
where the complex numbers $y,z\ne 0,1$, $y\ne z$,
So the four terms in \eqref{eq:w32} respectively have rank $8,4,4$ and $4$. 
The lower bound in \eqref{brnkw3} follows from \cite{Zuiddam2017A}.
We have proved our theorem.
\end{proof} 

Using Theorems \ref{thm:w2} and \ref{ubrankw3} we obtain that
\begin{eqnarray}
\label{eq:3m}
&&
\rank\ \cW^{\ox 3m}\le 20^m,
\\
\label{eq:3m+1}
&& 
\rank\ \cW^{\ox (3m+1)}\le 3\cdot 20^m,
\\ 	
\label{eq:3m+2}
&&
\rank\ \cW^{\ox (3m+2)}\le 8\cdot 20^m,
\end{eqnarray}
for any positive integer $m$. These equations give the upper bound of $\cW^{\ox n}$ for any positive integer $n$. On the other hand, a lower bound of $\rank \ox^n \cW$ is known as $\rank \ox_K^n \cW \ge 2^{n+1}-1$ \cite[Theorem 8]{ccd2010}.
It has been proved in the proof of Proposition 12 in \cite{cjz17} that $\rank\ \cW^{\ox n} \le (2n+1)2^n $. This upper bound is worse than  \eqref{eq:3m}-\eqref{eq:3m+2} for $n\in\{3,\ldots,9\}$ and better than  \eqref{eq:3m}-\eqref{eq:3m+2} for $n\ge 10$.   
\red{Using the above resuls we deduce that 
\[2(2n+1)^{\frac{1}{n}}\ge( \rank\ \ox^n\cW)^{\frac{1}{n}}\ge (\rank \ox_K^n \cW)^{\frac{1}{n}}\ge (2^{n+1}-1)^{\frac{1}{n}}.\]
Letting $n\to\infty$ we obtain}
\begin{equation}
\label{eq:lim}	
\lim_{n\rightarrow\infty} (\rank\ \ox^n\cW)^{1/n} =
\lim_{n\rightarrow\infty} (\rank\ \ox^n_K\cW )^{1/n} =2.
\end{equation}
In particular, the asymptotic rank is bounded above by border rank. This result has been also derived in \cite{Strassen1988The}.

Theorem \ref{ubrankw3}  shows that $8=\rank\ \otimes^2\cW> \rank\ \otimes_K^2\cW=7$.  Hence it is possible to assume that $\rank\ \otimes^3\cW>\rank\ \otimes_K^3\cW=16$.  The following lemma implies the above conjectured inequality under the following condition':
\begin{lemma}
\label{le:w3} If the 8-tensor $\cG(2,2)\otimes \cW^{\otimes 2}$ has rank 16, then $\rank\ \cW^{\ox3}>16$.
\end{lemma}
\begin{proof}  Assume by contradiction that
\[\cW^{\otimes 3}=\sum_{j=1}^{16} \otimes_{i=1}^9  \ba_{j,i}, \quad \ba_{j,i}\in\C^2 \textrm{ for } j\in[16], i\in[9].\]
Clearly, it is impossible that all $\ba_{j,1}$ are proportional to $\be_2$.  Without loss of generality we can assume that $\ba_{16,1}\not\propto \be_2$.   Let $\phi:\C^2 \to \C$ 
be nonzero linear functional such that $\phi(\ba_{16,1})=0$.  Let $\phi_k:\otimes^k\C^2\to \otimes^{k-1}\C$ be the linear transformations induced by $\phi$ for $k=3, 9$.  Recall that $\rank\ \phi_3(\cW)=2$.  So $\phi_3(\cW)$ is equivalent to the matrix $\cG(2,2)$.  Thus we obtain 
\[\phi_3(\cW)\otimes\cW^{\otimes 2}=\sum_{j=1}^{15} \phi(\ba_{j,1})\otimes_{i=2}^9  \ba_{j,i}.\]
 This equality contradicts our assumption that $\rank\ \cG(2,2)\otimes \cW^{\otimes 2}=16$.	
\end{proof}

\section{Open problems}
\label{sec:openprob}
It seems that many known results as $\rank\ \otimes_K^2\cW=7, \rank\ \otimes_K^3\cW=16$ and Theorem \ref{rankG2W} follow from the fact that we have a good number of results on the rank of $3$-tensors.  In fact, the proof of Theorem \ref{thm:w2} follows from Theorem \ref{rankG2W}.  
In the following set of open problems we ask a number of open problems which are basically related to extension of our or known results we used.
\begin{opprob}\label{openproblems}
\begin{enumerate}
\item Let $\be_1=(1,0)\trans=\ket{0},\be_2=(0,1)\trans=\ket{1}$ be a standard basis in $\C^2$.  Denote by $\cW_n=\sum_{i=1}^n \otimes^{i-1}\be_1\otimes \be_2\otimes^{n-i}\be_1$ be a symmetric tensor in $\otimes^n\C^2$.  (It is known that $\rank\ \cW_n=n$ \cite[Theorem 3]{ccd2010}.)  Is it true that $\rank\ \cG(k,d)\otimes_K\cW_n=kn$?
\item View $\cW\otimes\cW$ as a $6$-tensor in $\otimes_{i=1}^6 \bU_i$.  Let $\cX\in \C^4\otimes (\otimes^4\C^2)$ and $\cY\in (\otimes^2\C^4)\otimes(\otimes^2\C^2)$ be  $\cW\otimes\cW$ viewed as tensor on $(\bU_1\otimes\bU_4)\otimes(\otimes_{j=3,5,6}\bU_j)$ and $(\bU_1\otimes\bU_4)\otimes (\bU_2\otimes\bU_5)\otimes \bU_3\otimes \bU_6$ respectively.   Hence
\[8=\rank\ \cW\otimes \cW\ge \rank\ \cX\ge \rank\ \cY\ge \rank \cW\otimes_K\cW=7.\]
Is it true that $ \rank\ \cY=8$?
\item Is it true that $\rank\ \cG(2,2)\otimes \cW\otimes\cW=16$?
\item Let $d$ be an integer greater than three. One can easily generalize Kruskal's theorem to $d$-tensors by viewing $d$-tensor as $3$-tensors by grouping the factors in $\otimes_{i=1}^d \C^{n_i}$.  See for example \cite{friedland16}.   Do there exists better generalizations?
\item Can one have good generalizations of \eqref{eq:w32} for $\cW^{\otimes n}$ for $n>3$?
\end{enumerate}
\end{opprob}

\section*{Acknowledgments}

We thank Joseph M. Landsberg, Jeroen Zuiddam and Karol $\dot{Z}$yckowski for useful conversations. \red{We thank the referee for useful comments.}
LC was supported by Beijing Natural Science Foundation (4173076), the NNSF of China (Grant No. 11501024), and the Fundamental Research Funds for the Central Universities (Grant Nos. KG12001101, ZG216S1760 and ZG226S17J6).

\bibliography{chenfried}

\end{document}